\documentclass[11pt, fleqn, a4paper]{scrartcl}

\usepackage[english]{babel}
\usepackage[utf8]{inputenc}
\usepackage[T1]{fontenc} 

\usepackage{amsfonts}
\usepackage{amsmath}
\usepackage{caption}
\usepackage{comment}
\usepackage{environ}
\usepackage{scrlayer-scrpage}
\usepackage{float}
\usepackage{fontaxes}
\usepackage[lmargin=1in,rmargin=1in,bottom=1.0in,top=1.0in,twoside=False]{geometry}
\usepackage{graphics}
\usepackage{hyperref}
\usepackage{manyfoot}
\usepackage{microtype}
\usepackage{mweights}
\usepackage{nccfoots}
\usepackage{setspace}
\usepackage{totpages}
\usepackage{trimspaces}
\usepackage{upquote}
\usepackage{url}
\usepackage{xcolor}
\usepackage{xkeyval}
\usepackage{mathtools}
\usepackage{amssymb}
\usepackage{mathrsfs}

\definecolor{dark-blue}{rgb}{0.05,0.25,0.85}
\hypersetup{colorlinks=true,linkcolor=dark-blue,citecolor=dark-blue}

\usepackage[amsmath,thmmarks,hyperref]{ntheorem}
\usepackage[nameinlink]{cleveref}

\crefformat{page}{#2page~#1#3}%
\Crefformat{page}{#2Page~#1#3}%
\crefformat{equation}{#2(#1)#3}%
\Crefformat{equation}{#2(#1)#3}%
\crefformat{figure}{#2Figure~#1#3}%
\Crefformat{figure}{#2Figure~#1#3}%
\crefformat{section}{#2Section~#1#3}
\Crefformat{section}{#2Section~#1#3}
\crefformat{chapter}{#2Chapter~#1#3}
\Crefformat{chapter}{#2Chapter~#1#3}
\crefformat{chapter*}{#2Chapter~#1#3}
\Crefformat{chapter*}{#2Chapter~#1#3}
\crefformat{part}{#2Part~#1#3}
\Crefformat{part}{#2Part~#1#3}
\crefformat{enumi}{#2(#1)#3}
\Crefformat{enumi}{#2(#1)#3}
\crefformat{algorithm}{#2algorithm~#1#3}%
\Crefformat{algorithm}{#2Algorithm~#1#3}%

\usepackage{latexsym}

\theoremnumbering{arabic}
\theoremstyle{plain}
\theoremsymbol{}
\theorembodyfont{\itshape}
\theoremheaderfont{\normalfont\bfseries}
\theoremseparator{}

\theoremseparator{.}
\newtheorem{theorem}{Theorem}
\crefformat{theorem}{#2Theorem~#1#3}
\Crefformat{theorem}{#2Theorem~#1#3}

\newcommand{\newtheoremwithcrefformat}[2]{%
  \newtheorem{#1}[theorem]{#2}%
  \crefformat{#1}{##2\MakeUppercase#1~##1##3}%
  \Crefformat{#1}{##2\MakeUppercase#1~##1##3}%
}

\newtheoremwithcrefformat{proposition}{Proposition}
\newtheoremwithcrefformat{observation}{Observation}
\newtheoremwithcrefformat{lemma}{Lemma}
\newtheoremwithcrefformat{conjecture}{Conjecture}
\newtheoremwithcrefformat{corollary}{Corollary}
\theorembodyfont{\upshape}
\newtheoremwithcrefformat{example}{Example}
\newtheoremwithcrefformat{remark}{Remark}

\newtheoremwithcrefformat{definition}{Definition}

\theoremsymbol{\ensuremath{\square}}
\theoremstyle{nonumberplain}
\theoremseparator{}
\theoremheaderfont{\scshape}
\theorembodyfont{\normalfont}
\theoremsymbol{\ensuremath{\square}}
\newtheorem{proof}{Proof.}

\newcommand{\sref}[2]{\hyperref[#2]{#1~\ref{#2}}}

\newcommand{\wcol}{\mathrm{wcol}}
\newcommand{\col}{\mathrm{col}}

\newcommand{\td}{\mathrm{td}}

\newcommand{\WReach}{\mathrm{WReach}}

\newcommand{\str}{\mathbb}

\newcommand{\Oof}{\mathcal{O}}

\newcommand{\Cc}{\mathscr{C}}

\newcommand{\XXX}{\mathcal{X}}
\newcommand{\Xx}{\mathcal{X}}
\newcommand{\minor}{\preccurlyeq}
\newcommand{\N}{\mathbb{N}}
\newcommand{\R}{\mathbb{R}}
\newcommand{\dist}{\mathrm{dist}}
\newcommand{\rad}{\mathrm{rad}}

\usepackage{tikz}
\usetikzlibrary{decorations.pathreplacing}
\tikzstyle{vertex}=[circle,inner sep=1.5,minimum size =1.5mm,semithick,fill=black, draw=black]

\title{Nowhere dense graph classes and \\algorithmic applications\\ {\large A tutorial at Highlights of Logic, Games and Automata 2019}}
\author{
Sebastian Siebertz
\\University of Bremen\\
\texttt{siebertz@uni-bremen.de}}
\date{}

\begin{document}
\maketitle

\vspace{-1.1cm}
\begin{abstract}
\noindent \textbf{Abstract.} The notion of nowhere dense graph classes 
was introduced by Ne\v{s}et\v{r}il and
Ossona de Mendez and 
provides a robust concept of uniform sparseness of graph classes. 
Nowhere dense classes generalize many familiar classes of sparse graphs such as classes that exclude a fixed graph as a minor or topological minor. 
They admit several seemingly unrelated 
natural characterizations that lead to strong algorithmic applications.
In particular, the model-checking
problem for first-order logic is fixed-parameter tractable over 
these classes. These notes, prepared for a tutorial at Highlights of
Logic, Games and Automata 2019, are a brief introduction to the theory 
of nowhere denseness, driven by algorithmic applications.
\end{abstract}

\section{Introduction}

The notion of excluded minors is celebrated as one of the most
successful notions in contemporary graph theory and has an immense influence on algorithmic graph theory. At its heart 
lies the structure theorem that states that every graph $G$ that
excludes a fixed graph $H$ as a minor 
can be decomposed in a treelike way
into parts that can be almost topologically embedded on a surface that~$H$ does not embed on~\cite{robertson1999graph}. Surprisingly, the theory of
\emph{bounded expansion} and 
\emph{nowhere dense graph classes}, i.e.\ the theory of 
\emph{bounded depth minors}, 
which is much simpler and yet deals with much more general graph classes,  
is much less known. The notions of bounded 
expansion and nowhere denseness 
were introduced by Ne\v{s}et\v{r}il and Ossona de Mendez~\cite{nevsetvril2008grad, nevsetvril2011nowhere} and 
provide a robust concept of uniform sparseness of graph classes. 
Classes with bounded expansion and nowhere dense classes generalize many familiar classes 
of sparse graphs, such as classes that exclude a fixed 
graph as a minor or topological minor. 
They admit several %seemingly unrelated 
natural characterizations that lead to strong algorithmic applications.
In particular, the model-checking
problem for first-order logic is fixed-parameter tractable over 
these classes~\cite{dvokratho10, grohe2017deciding}. 

In this short exposition I would like to give a very
accessible introduction to the theory of bounded
expansion and nowhere denseness. The presentation is 
driven by the application of solving the first-order 
model-checking problem based on Gaifman's
locality theorem. Therefore, I focus on the aspect 
of appropriately localizing well known 
width measures from graph theory. The original 
definitions of bounded expansion and nowhere dense
classes are given by imposing restrictions on the 
\emph{bounded depth minors} that can be found 
in graphs from the class. An algorithmically very useful
equivalent definition of nowhere dense classes is given 
in terms of \emph{uniform quasi-wideness}, which is 
often considered as one of the more cumbersome
parts of the theory. I will present this concept as a local
version of treewidth and as a local version of treedepth, 
and hope to convince the reader of the beauty of the 
concept. Finally, a third characterization is provided in 
terms of weak reachability numbers, which can again 
be seen as a local version of treedepth.

\section{First-order model-checking}

First-order logic can express many interesting algorithmic 
properties of graphs such as the existence of an independent
set of size at least $k$, the existence of a (connected) 
dominating set of size at most $k$, and many more. 
For example, the formula \[\varphi_{\mathrm{dom}}^k\coloneqq 
\exists x_1\ldots\exists x_k\forall y\big(\bigvee_{1\leq i\leq k}
(y=x_i \vee E(y,x_i)\big)\] is true in a graph $G$ if and only
if $G$ has a dominating set of size at most $k$.  
The model-checking problem for 
first-order logic is the problem to test for an input structure
$\str A$ and input formula~$\varphi$ whether~$\varphi$ holds
in $\str A$, in symbols $\str A\models \varphi$. 

The model-checking problem for first-order logic on an 
input structure $\str A$ of order $n$ is decidable
in time $n^{\Oof(q)}$, where $q$ is the quantifier 
rank of the formula $\varphi$. Phrased in terms of parameterized
complexity, the problem belongs to the complexity 
class \textsf{XP} of \emph{slicewise polynomial} problems. 
It is expected that in general this running time cannot be 
avoided, e.g.\ testing whether a graph contains a clique
with $k$ vertices, which requires exactly $k$ quantifiers, 
cannot be done in time~$n^{o(k)}$ unless the exponential
time hypothesis fails~\cite{chen2006strong}. 

This has led to the investigation of structural properties 
of the input structures, especially of graphs, that allow
for more efficient model-checking. In particular, we search
for (the most general) graph classes on which the problem
can be solved in time $f(|\varphi|)\cdot n^c$ for some 
computable function~$f$ and constant $c$, i.e.\ for
classes where the problem is \emph{fixed-parameter tractable}
(parameterized by the formula length). As tractability 
for the model-checking problem 
of a logic implies tractability not
only for individual problems but for whole classes 
of problems, a tractability result for a model-checking
problem is often referred to as an \emph{algorithmic 
meta theorem}. 
There is a long line of meta theorems for first-order logic on
sparse structures~\cite{dawgrokre07,dvokratho10,flugro01,frigro01,
  kre11,see96}, culminating in the result that the model-checking
  problem for first-order logic is fixed-parameter
  tractable on every nowhere dense class of graphs~\cite{grohe2017deciding}. It was shown earlier that on 
every subgraph-closed graph class that is not nowhere dense
the problem is as hard as on all graphs~\cite{dvokratho10, kre11}, hence, the classification
of tractability for the first-order model-checking problem on subgraph-closed classes is 
essentially complete. 

\smallskip
The key property 
of first-order logic that is exploited for efficient model-checking
is \emph{locality}. Gaifman's Locality Theorem states that every
first-order formula~$\varphi(\bar x)$ is equivalent to a Boolean combination of 
\begin{enumerate}
\item \emph{local formulas} $\psi^{(r)}(\bar x)$ and 
\item \emph{basic local 
formulas} $\exists x_1\ldots \exists x_k\big(\bigwedge_{i\neq j}
\dist(x_i,x_j)>2r \wedge \chi^{(r)}(x_i)\big)$. 
\end{enumerate} 

Here, the notation $\psi^{(r)}(\bar x)$ means that for every
graph $G$ and every tuple $\bar v\in V(G)^{|\bar x|}$ we have 
$G\models \psi^{(r)}(\bar v)$  if and only if $G[N_r(\bar v)]\models \psi^{(r)}(\bar v)$, 
where $G[N_r(\bar v)]$ denotes the subgraph of~$G$ induced
by the $r$-neighborhood $N_r(\bar v)$ of $\bar v$. This 
property is syntactically ensured in $\psi^{(r)}$ 
by relativizing all quantifiers to distance at most $r$ from one 
of the the free variables. The numbers~$r$ and $k$ in the 
formulas above depend only on the formula~$\varphi$, 
and furthermore, the Gaifman normal form of any formula $\varphi$
is computable from $\varphi$. 

This translates the model-checking problem to the following 
algorithmic problem. To decide for a graph $G$ and tuple 
$\bar v\in V(G)^{|\bar x|}$ whether $G\models \varphi(\bar v)$, 
\begin{enumerate}
\item decide whether $\bar v$ has the local properties described
by $\psi^{(r)}(\bar x)$;
\item decide for each $v\in V(G)$ whether $G\models \chi^{(r)}(v)$;
\item solve each \emph{generalized
independent set} problem described by the basic local formulas 
$\exists x_1\ldots \exists x_k\big(\bigwedge_{i\neq j}
\dist(x_i,x_j)>2r \wedge \chi^{(r)}(x_i)\big)$, and finally
\item evaluate the Boolean combination of these 
statements that is equivalent to $\varphi$. 
\end{enumerate}  

\section{Bounded depth minors, bounded 
expansion and nowhere denseness}

By Gaifman's theorem, we expect to have efficient 
model-checking algorithms on graph classes that  
locally have nice properties. With this motivation in mind
we can try to find appropriate local versions of width 
measures that we know how to handle well. This approach
was followed e.g.\ in \cite{dawgrokre07,frigro01}
where it was simply required that the $r$-neighborhoods in graphs
from the class have good properties, e.g.\ they have bounded
treewidth, or exclude a minor. For 
example, we say that a class $\Cc$ has \emph{locally bounded 
treewidth} if for every $r\in \N$ there exists a number 
$t=t(r)$ such that for every $G\in \Cc$ and 
every $v\in V(G)$
the treewidth of $G[N_r(v)]$ is bounded
by $t$. 
%Here, $N_r(v)$ denotes the $r$-neighborhood 
%of the vertex $v$ in $G$. 
Similarly, we say that a class
$\Cc$ \emph{locally excludes a minor} if for every $r\in \N$ 
there exists a number $m=m(r)$ such that for every $G \in \Cc$ 
and every $v\in V(G)$  the graph $G[N_r(v)]$
excludes the complete graph~$K_m$ on $m$ vertices as a minor (the concepts of 
treewidth and minors are defined formally below). Note however, that 
this approach of defining locally well behaved classes
is not very robust. For example, if we add to every 
graph $G\in \Cc$ an 
apex vertex, i.e.\ a vertex that is connected with every
other vertex of~$G$, then the resulting class has locally
bounded treewidth if and only if the class $\Cc$ has 
bounded treewidth. On the other, it is very easy to 
algorithmically handle the apex vertices and we are looking
for more robust locality notions. 

The following notion of \emph{bounded depth 
minors} is the fundamental notion in the theory 
of bounded expansion and nowhere denseness~ \cite{nevsetvril2008grad,nevsetvril2011nowhere}. 

\begin{definition}
A graph $H$ is a \emph{minor} of $G$, written $H\minor G$, if there is 
a map $\phi$ that assigns to every vertex $v\in V(H)$ a
connected subgraph $\phi(v) \subseteq G$ of $G$ and to every edge
$e\in E(H)$ an edge  $\phi(e)\in E(G)$ such that
\begin{enumerate}
  \item if  $u,v\in V(H)$ with $u\not= v$, then $V(\phi(v)) \cap V(\phi(u)) =
    \emptyset$ and
  \item if $e = uv \in E(H)$, then $\phi(e) = u'v'\in E(G)$ for some
   $u'\in V(\phi(u))$ and $v'\in V(\phi(v))$.
\end{enumerate}
The set $\phi(v)$ for a vertex $v\in V(H)$ is called the \emph{branch
set} or \emph{model of $v$ in $G$}. 
The map~$\phi$ is called the \emph{minor model of $H$ in $G$}. 
The \emph{depth} of a minor model is the maximal radius of its branch 
sets. For $r\in \N$, the graph $H$ is a {\em{depth-$r$ minor}} of $G$, written $H\minor_rG$, if there is a minor model
$\phi$ of~$H$ in $G$ of depth at most $r$. 
%such that each branch set $\phi(v)$ for
%$v\in V(H)$ has radius at most $r$. 
%
%If $\Cc$ is a class of graphs and $r$ is a positive integer we
%write $\Cc\mathop{\triangledown}r$ for the class of all
%depth-$r$ minors of graphs from $\Cc$. 
\end{definition}

Now, bounded expansion and nowhere dense classes are defined 
by imposing restrictions on the structure of bounded depth minors.

%\begin{definition}
%Let $G$ be a graph. The \emph{edge density} of $G$ is
%\[\nabla(G)=|E(G)|/|V(G)|.\] For $r\in \N$,
%the \emph{greatest reduced average density}~$\nabla_r(G)$  of $G$ with rank~$r$ is 
%\[\nabla_r(G)\coloneqq \sup\left\{\frac{|E(H)|}{|V(H)|}\ \colon\ H\minor_r G\right\}.\]
%and the \emph{greatest topological reduced average density}~$\widetilde{\nabla}_r(G)$  of $G$ with rank~$r$ is 
%\[\widetilde{\nabla}_r(G)\coloneqq \sup\left\{\frac{|E(H)|}{|V(H)|}\ \colon\ H\minor_r^{top} G\right\}\]
%Similarly, we define 
%\[\omega_r(G)\coloneqq \sup \{t\ \colon\ K_t\minor_r G\}\]
%and 
%\[\widetilde{\omega}_r(G)\coloneqq \sup \{t\ \colon\ K_t\minor_r^{top} G\}.\]
%%
%These notions  extend to classes $\Cc$ of graphs
%by $\nabla_r(\Cc)\coloneqq\sup_{G\in\Cc}\nabla_r(G)$, 
%$\widetilde{\nabla}_r(\Cc)\coloneqq\sup_{G\in\Cc}\widetilde\nabla_r(G)$, 
%$\omega(\Cc)\coloneqq\sup_{G\in\Cc}\omega_r(G)$, and
%$\widetilde{\omega}_r(\Cc)\coloneqq\sup_{G\in\Cc}\widetilde\omega_r(G)$.
%\end{definition}

%We now come to the definition of the notions of bounded expansion
%and nowhere denseness. 
 
\begin{definition}
 A class $\Cc$ of graphs has
\emph{bounded expansion} if for every $r\in \N$ there exists 
number $d=d(r)$ such that  the edge density $|E(H)|/|V(H)|$ of
every $H\minor_r G$ for $G\in \Cc$ 
is \mbox{bounded by $d$}. %we have $\nabla_r(\Cc)\leq d(r)$. 
%The class $\Cc$ has \emph{polynomial expansion} if the function $d$ is polynomially bounded.  
\end{definition}

\begin{definition}
A class $\Cc$ of graphs is \emph{nowhere dense} if for every $r\in \N$ there exists a number $m=m(r)$ such that  we have 
$K_m\not\minor_r G$ for all $G\in\Cc$. 
%$\omega_r(\Cc)\leq t(r)$. 
\end{definition}

\begin{example*}
\mbox{}
\begin{enumerate}
\item Every class $\Cc$ that excludes a fixed graph $H$ as
a minor has bounded expansion. For such classes there exists an absolute constant
$c$ such that for all $r\in \N$ 
the edge density of depth-$r$ minors of graphs in $\Cc$
is bounded by $c$. 
Special cases are classes of bounded treewidth, the class of
planar graphs, and 
every class of graphs that can be drawn 
with a bounded number of crossings, see~\cite{nevsetvril2012characterisations}, and every class of graphs
that embeds into a fixed surface. 
\item Every class $\Cc$ that excludes a fixed graph $H$ as
a topological minor has bounded expansion. 
%For such classes there exists an absolute constant
%$c$ such that for all $G\in\Cc$ and all positive integers~$r$ 
%we have $\widetilde\nabla_r(G)\leq c$. 
Every class that 
excludes $H$ as a minor also excludes $H$ as a topological 
minor. Further special cases are classes of bounded degree
and classes of graphs that can be drawn 
with a linear number of crossings, see~\cite{nevsetvril2012characterisations}. 
\item Every class of graphs that can be drawn with a bounded 
number of crossings per edge has bounded expansion~\cite{nevsetvril2012characterisations}. 
\item Every class of graphs with bounded queue-number, 
bounded stack-number or bounded non-repetitive chromatic number has bounded expansion~\cite{nevsetvril2012characterisations}. 
\item The class of Erd\"os-R\'enyi random graphs with 
constant average degree $d/n$, $G(n,d/n)$, has 
asymptotically almost surely bounded expansion~\cite{nevsetvril2012characterisations}. 
\item Every bounded expansion class is nowhere dense.
\item The class of graphs with girth greater than maximum degree
is nowhere dense (and has locally bounded treewidth) and does not have bounded expansion~\cite{nevsetril2012sparsity}. 
\end{enumerate}
\end{example*}

Nowhere dense classes can also be defined in terms of 
\emph{subdivisions} or \emph{topological minors}. This fact
is very useful for proving algorithmic
lower bounds for classes that are not nowhere dense. 
For $r\in \N$, a graph $H$ is an \emph{$r$-subdivision} of a 
graph $G$ if $H$ is obtained from~$G$ by replacing
every edge by a path of length $r+1$ (containing
$r$ inner vertices). 
%Using a simple Ramsey argument
%we get the following corollary. 

\begin{lemma}\label{cor:sd}
Let $\Cc$ be a class that is not nowhere dense
and that is closed under taking subgraphs. Then 
there exists $r\in \N$ such that $\Cc$ contains an 
$r$-subdivision of every graph $H$. 
\end{lemma}

Using the lemma it is for example not difficult to show
that the first-order model-checking problem on every
class that is not nowhere dense and closed under taking
subgraphs is as hard as on the class of all graphs. 

Finally, we note that 
nowhere dense classes are sparse. 

\begin{theorem}[\cite{dvovrak2007asymptotical,nevsetvril2011nowhere}]
A class 
$\Cc$ of graphs is nowhere dense if and only 
if for all real
$\epsilon>0$ and all $r\in \N$ there exists 
an integer~$n_0$ such that all
$n$-vertex graphs $H\minor_r G$ for $G\in \Cc$ 
with $n\geqslant n_0$ have edge density at most 
$n^\epsilon$.  
\end{theorem}

At this point the notions of bounded expansion and
nowhere denseness are established as abstract
concepts. Observe that we have achieved the desired
robustness of the concepts under small changes, such
as adding apex vertices to the graphs of a class $\Cc$. 
On the other hand observe that we cannot expect
to find a structure theorem as for classes that 
exclude a fixed minor $H$. For example the class
of graphs that contains the $n$-subdivision of every
$n$-vertex graph $G$ has bounded expansion and 
we cannot find a global decomposition for the graphs
from this class. This example also shows the 
limitations for algorithmic applications. We will e.g.\ not
be able to solve \emph{global} connectivity problems more efficiently
than on general graph classes. 
We will now move to the tools that can be used
to handle bounded expansion and nowhere dense
graph classes. 
%One can understand
%them as local versions of treedepth and treewidth. 
%To show the equivalence of the concepts we need
%to study also the concept of uniform quasi-wideness, 
%which is often considered as one of the more 
%cumbersome parts of the theory. I hope that the 
%given applications can convince the reader of the
%beauty of the concept. 

\section{Uniform quasi-wideness and separating
neighborhoods}

The \emph{separator width} of a graph $G$ is defined 
as the minimum number $k$ such that for every 
$A\subseteq V(G)$ there exists a set $S$ of order 
at most $k$ such that for every component $C$ of 
$G-S$ we have $|V(C)\cap A|\leq |A|/2$. A class of graphs
has bounded treewidth if and only if it has bounded 
separator width. In fact, the main algorithmic applications 
of graphs with bounded treewidth follow from the property 
that these graphs admit small balanced
separators. Following our goal of 
finding an appropriate localization of this property
we give the following definition. 

\begin{definition}\label{def:neighborhood-sep}
A class $\Cc$ of graphs admits \emph{balanced neighborhood
separators} if for every $r\in \N$  and every real $\epsilon>0$ 
there exists a number $s=s(r,\epsilon)$ such that the following holds. For every graph $G\in\Cc$ and every subset $A\subseteq V(G)$ 
there exists a set $S\subseteq V(G)$ of 
order at most $s$ such that the $|N_r(v)^{G-S}\cap A|\leq \epsilon|A|$ 
for all $v\in V(G)\setminus S$. 
\end{definition}

\begin{theorem}[\cite{nevsetvril2016structural}]\label{thm:neighborhood-sep}
A class $\Cc$ of graphs is nowhere dense if 
and only if $\Cc$ admits balanced neighborhood 
separators. 
\end{theorem}

The proof uses the above mentioned characterization 
of nowhere dense classes in terms of uniform
quasi-wideness that we define next. 
If $G$ is a graph and $A\subseteq
V(G)$, then $A$ is \emph{distance-$r$ independent} 
if the vertices of $A$ have pairwise distance greater
than $r$ in $G$. 

\begin{definition}\label{def:uqw}
  A
  graph class~$\Cc$ is {\em{uniformly quasi-wide}} 
  if for all $r,m\in \N$ there exist numbers $s=s(r)$ and 
  $N=N(r,m)$ such that the following holds. For every graph
  $G\in \Cc$ and every set $A\subseteq V(G)$:
  if $|A|\geq N$, then there exists
  $S\subseteq V(G)$ with $|S|\leq s$ and $B\subseteq A\setminus S$ with $|B|\geq m$ such that 
  $B$ is distance-$r$ independent in $G-S$. 
\end{definition}

\begin{theorem}[\cite{nevsetvril2011nowhere}]\label{thm:nd-uqw}
A class $\Cc$ of graphs is nowhere dense if and only
if $\Cc$ is uniformly quasi-wide. 
\end{theorem}

We are not going to prove \cref{thm:nd-uqw} as the proof
is quite technical. However, to get familiar with the concept
of uniform quasi-wideness it is instructive to 
prove \cref{thm:neighborhood-sep}. 

\begin{proof}[of \cref{thm:neighborhood-sep}]
Let $\Cc$ be nowhere dense and let $r\in \N$ and 
$\epsilon>0$. According to \cref{thm:nd-uqw},
$\Cc$ is uniformly quasi-wide. Hence, for $r'=4r$ there 
exists $s=s(r')$ and for $m=\left\lfloor 1/\epsilon\right\rfloor+s+1$ there exists $N=N(r',m)$ such that 
for every graph $G\in \Cc$ and every $X\subseteq V(G)$:
  if $|X|\geq N$, then there exists
  $Y\subseteq V(G)$ with $|Y|\leq s$ and $X'\subseteq X\setminus Y$ with $|X'|\geq m$ such that 
  $X'$ is distance-$4r$ independent in $G-Y$. 

Let $G\in \Cc$ and $A\subseteq V(G)$. 
We aim to prove that there exists $S\subseteq V(G)$ with 
$|S|\leq N$ such that $|N_r^{G-S}(v)\cap A|\leq \epsilon|A|$ 
for all $v\in V(G)\setminus S$. 

Let $X\subseteq V(G)$ be any set such that for all $v\in V(G)\setminus X$ we have 
$|N_r(v)^{G-X}\cap A|\leq \epsilon |A|$. We show that 
if $|X|>N$, then there exists $Z\subseteq V(G)$ with 
$|Z|<|X|$ such that for all $v\in V(G)\setminus Z$ 
we have $|N_r^{G-Z}(v)\cap A|\leq \epsilon |A|$. The claim
follows by repeating the argument until $|Z|\leq N$ and 
then setting $S=Z$. 

If $|X|>N$, then there 
exist sets $Y\subseteq V(G)$ with 
$|Y|\leq s$ and $X'\subseteq X\setminus Y$ with $|X'|\geq 
\left\lfloor 1/\epsilon\right\rfloor+s+1$ that is 
distance-$4r$ independent in $G-Y$. In $X'$ there are at most 
$\left\lfloor 1/\epsilon\right\rfloor$ vertices $v$ with 
$|N_{2r}^{G-Y}(v)\cap A|\geq \epsilon|A|$, as these neighborhoods are 
disjoint. Hence, since $|X'|\geq 
\left\lfloor 1/\epsilon\right\rfloor+s+1$, 
there is a subset $X''\subseteq X'$ with 
$|X''|>s$ and such that $|N_{2r}^{G-Y}(v)\cap A|\leq \epsilon |A|$ for all $v\in X''$. Let $Z=(X\setminus X'')\cup Y$. 

We claim that every $v\in V(G)\setminus Z$ satisfies 
$|N_r(v)^{G-Z}\cap A|\leq\epsilon |A|$. 
To see this, let $v\in V(G)\setminus Z$ and assume 
$N_r^{G-Z}(v)\cap A \neq N_r^{G-X}(v)\cap A$ (we have
to consider only such elements, as by assumption $|N_r^{G-X}(v)\cap A|\leq \epsilon|A|$
for all $v\in V(G)\setminus X$). 
This implies that there is $x\in X''$ such that
$\mathrm{dist}_{G-Z}(v,x)< r$. 
This implies $N_r^{G-Z}(v)\subseteq N_{2r}^{G-Z}(x)$. 
By construction we have $|N_{2r}^{G-Z}(x)\cap A|\leq \epsilon |A|$, 
as claimed. %This finishes the proof of the forward direction of 
%the theorem. 

\medskip
Vice versa, assume $\Cc$ is not nowhere dense. 
We show that $\Cc$ does not admit balanced
neighborhood covers. Let $s\colon \N\times\R\rightarrow\N$ be an arbitrary function. 
According to \cref{cor:sd}, there is $r\in \N$ such that
$\Cc$ contains an $r$-subdivision of every graph 
$H$. Let $n\coloneqq 2s(2r,1/2)$. Let $G\in \Cc$ be 
such that an $r$-subdivision of $K_n$ is a subgraph
of $G$. Let $A$ be a set of vertices of $G$ that contains
the vertices of this subdivision. Let $S$ be any set
of size at most $s(2r,1/2)$. Then the graph $G[A\setminus S]$
contains a vertex whose $2r$-neighborhood has order 
at least $|A|-n-\binom{n/2}{2}r>|A|/2$. Hence, $s$ is not a function for choosing $s=s(r,\epsilon)$ for balanced neighborhood
separators. As $s$ was chosen arbitrary, this proves
the claim. 
\end{proof}

The proof of \cref{thm:neighborhood-sep} can be 
made algorithmic: we algorithmically iterate
the exchange argument of the proof until we arrive at a 
set of order at most $N$. In each 
step we need to compute a set $Y$ (the set $S$ in the
definition of uniform quasi-wideness). This can be done
in polynomial time, see~\cite{kreutzer2018polynomial, pilipczuk2018number}. The work~\cite{pilipczuk2018number} gives also the best known bounds for the function~$N$
in the definition of uniform quasi-wideness. 

The algorithmic applications lie at hand. We can 
recursively decompose local neighborhoods into 
smaller and smaller pieces such that the recursion
stops after $\log n$ steps. In the next section we
will see that we can do even better and get a recursion
tree of  depth depending only on $r$.  

\section{Uniform quasi-wideness and splitting
neighborhoods}

Graph classes whose members admit tree 
decompositions of bounded width and bounded 
depth are called classes with \emph{bounded treedepth}. 
The notion of treedepth was introduced by Ne\v{s}et\v{r}il and
Ossona de Mendez 
in~\cite{nevsetvril2006tree}, equivalent notions were 
studied before under different names. We refer to~\cite{nevsetril2012sparsity} for a discussion 
on the various equivalent parameters. 

A \emph{rooted tree $T$} is an acyclic connected graph with 
one designated root vertex.
This imposes the standard \emph{ancestor/descendant
relation} in $T$: 
a node $v$ is a descendant of all the nodes that appear on the unique 
path leading from $v$ to the root. 
A \emph{rooted forest} $F$ is a disjoint union of rooted trees. 
We write $u\leq_F v$ if~$u$ is an ancestor of $v$ in $F$.
The relation~$\leq_F$ is a partial order on the nodes of $F$ with the 
roots being the $\leq_F$-minimal elements. 
The \emph{depth} of a vertex $v$ in a rooted forest $F$ 
is the number of vertices on the path from $v$ to the root 
(of the tree to which $v$ belongs).  The \emph{depth} 
of $F$ is the maximum depth of the vertices of $F$.

\begin{definition}
Let $G$ be a graph. The \emph{treedepth} $\td(G)$ of $G$
is the minimum depth of a rooted forest $F$ on the 
same vertex set as $G$ such that 
whenever $uv\in E(G)$, then $u\leq_F v$ or $v\leq_F u$. 
\end{definition}

We can equivalently define treedepth by the following elimination 
game. Let $\ell\in \N$. The \emph{$\ell$-round treedepth
game} on a graph $G$ is played by two players, 
\emph{connector} and \emph{splitter}, as follows. 
We let~$G_0:=G$. In round~$i+1$ of the game, connector chooses 
a component $C_{i+1}$ of $G_i$. Then splitter picks a vertex
$w_{i+1}\in V(C_{i+1})$. We let~$G_{i+1}:=C_{i+1}-\{w_{i+1}\}$. 
Splitter wins if~$G_{i+1}=\emptyset$. Otherwise the game 
continues at~$G_{i+1}$. If splitter has not won after~${\ell}$ 
rounds, then connector~wins.

A \emph{strategy} for splitter is a function~$\sigma$ that maps every
partial play $(C_1, w_1, \dots,C_s, w_s)$, with associated
sequence~$G_0, \dots, G_s$ of graphs, and the next
move~$C_{s+1}$ of connector, to a vertex \mbox{$w_{s+1}\in V(C_{s+1})$} that is the next move of
splitter. A strategy~$\sigma$ is a \emph{winning strategy} for
splitter if splitter wins every play in which she follows the
strategy~$f$. We say that splitter \emph{wins} the simple $\ell$-round
radius-$r$ splitter game on~$G$ if she has a winning strategy.

\begin{lemma}[Folklore]
A graph $G$ has treedepth $\ell$ if and only if splitter wins the
$\ell$-round treedepth game on $G$. 
\end{lemma}

We now consider the following change of the rules of the game that is motivated by our goal to find an appropriate localization of treedepth. 
The game gets an additional parameter $r$ for the 
radius. 
Instead of picking in round $i+1$ of the game a component $C_{i+1}$ 
of the currently considered graph $G_i$, connector picks a 
subgraph $C_{i+1}$ of radius at most $r$ in $G_i$. Formally, we consider the following game.

 Let~${\ell},r\in \N$. The \emph{simple $\ell$-round
  radius-$r$ splitter game} on a graph~$G$ is played by two players,
\emph{connector} and \emph{splitter}, as follows. We let~$G_0:=G$. In
round~$i+1$ of the game, connector chooses a
subgraph $C_{i+1}$ of $G_i$ of radius at most $r$. 
Then splitter picks a vertex
$w_{i+1}\in V(C_{i+1})$. We
let~$G_{i+1}:=C_{i+1}-\{w_{i+1}\}$. Splitter
wins if~$G_{i+1}=\emptyset$. Otherwise the game continues
at~$G_{i+1}$. If splitter has not won after~${\ell}$ rounds, then
connector wins. Strategies are defined as above. 

%A \emph{strategy} for splitter is a function~$\sigma$ that maps every
%partial play $(C_1, w_1, \dots,C_s, w_s)$, with associated
%sequence~$G_0, \dots, G_s$ of graphs, and the next
%move~$C_{s+1}$ of connector, to a
%vertex \mbox{$w_{s+1}\in V(C_{s+1})$} that is the next move of
%splitter. A strategy~$\sigma$ is a \emph{winning strategy} for
%splitter if splitter wins every play in which she follows the
%strategy~$f$. We say that splitter \emph{wins} the simple $\ell$-round
%radius-$r$ splitter game on~$G$ if she has a winning strategy.

\begin{theorem}[\cite{grohe2017deciding}] A class $\Cc$
  of graphs is nowhere dense if and only if for every $r\in \N$ 
  there exists a number $\ell=\ell(r)$ such that 
  splitter wins the simple
  $\ell$-round radius-$r$ splitter game on every graph $G\in\Cc$.
\end{theorem}
\begin{proof}
For convenience we allow splitter in every round $i$
to delete not a single vertex $w_i$ 
but a set $W_i$ of $m(r)$ vertices
for any fixed function $m$. Obviously this does not
give him additional power, as he can simulate the 
deletion of $m$ vertices in $m$ rounds of the game. 

  Let $r\in \N$. As~$\Cc$ is nowhere dense, it is also uniformly
  quasi-wide. Let~$s=s(r)$ and \mbox{$N=N(r,2s+2)$} be the 
  numbers satisfying the properties of \cref{def:neighborhood-sep}. 
  Let~${\ell}:=N$ and \mbox{$m(r)\coloneqq
  {\ell}\cdot (r+1)$}. Note that both~${\ell}$ and~$m$ only depend
  on~$\Cc$ and~$r$. We claim that for any~$G\in\Cc$, 
  splitter wins the $\ell$-round radius-$r$ splitter game
  in which splitter is allowed to delete~$m(r)$ vertices in each round. 

  Let~$G\in\Cc$ be a graph. In the game
  on~$G$, splitter uses the following strategy. In the first round, if
  connector chooses a subgraph $C_1$ of $G_0=G$ of radius
  at most $r$, say rooted at a vertex $v_1\in V(C_1)$, i.e.\ 
  $V(C_1)\subseteq N_r(v_1)$, 
  then splitter
  chooses the set~$W_1:=\{v_1\}$. Now let~$i>1$ and suppose that~$v_1,\ldots,
  v_i, G_1,\ldots, G_i, W_1,\ldots, W_i$ have already been
  defined. Suppose connector chooses a subgraph~$C_{i+1}$ 
  of $G_i$, say rooted at $v_{i+1}\in V(G_i)$. We
  define~$W_{i+1}$ as follows. For each~$1\leq j\leq i$, choose a
  path~$P_{j,i+1}$ in~$C_j$ of length at
  most~$r$ connecting~$v_j$ and~$v_{i+1}$. Such a path 
  must exist
  as~$v_{i+1}\in V(C_{i})\subseteq V(C_j)\subseteq
  N_r^{G_{j-1}}(v_j)$. We let~$W_{i+1}:=\bigcup_{1\leq j\leq
    i}V(P_{j,i+1})\cap V(C_{i+1})$. Note that~$|W_{i+1}|\leq
  i\cdot (r+1)$ (the paths have length at most~$r$ and hence consist
  of~$r+1$ vertices). It remains to be shown that the length of any
  such play is bounded by ~${\ell}$.
  
  Assume towards a contradiction that connector can survive on~$G$
  for~${\ell}'={\ell}+1$ rounds. \mbox{Let~$(v_1,\ldots,
  v_{{\ell}'},$ $G_1,\ldots, G_{{\ell}'}, W_1,\ldots, W_{{\ell}'})$} be
  the play. As~${\ell}'>N(r,2s+2)$, for~$W:=\{v_1,\ldots,
  v_{{\ell}'}\}$ there is a set~$S\subseteq V(G)$
  with~$|S|\leq s$, such that~$W$ contains an~$r$-independent
  set~$I$ of size~$t:=2s+2$ in~$G-S$.   
  Without loss of generality assume that $I=\{v_1,\ldots, v_{\ell'}\}$.

  We now consider the pairs~$(v_{2j-1}, v_{2j})$ for~$1\leq j\leq
  s+1$. By construction,~$P_j:=P_{2j-1, 2j}$ is a path of
  length at most~$r$ from~$v_{2j-1}$ to~$v_{2j}$
  in~$G_{2j-2}$. Any path~$P_j$ must necessarily contain a
  vertex~$s_j\in S$, as otherwise the path would exist in~$G-S$,
  contradicting the fact that~$I$ is \mbox{$r$-independent} in~$G - S$. 
  We claim
  that for~$i\neq j$,~$s_i\neq s_j$, but this is not possible, as
  there are at most~$s$ vertices in~$S$. To prove the claim, 
  assume~$i>j$. Then~$V(P_j)\cap V(G_{2j-1})\subseteq W_{2j}$, 
  thus~$V(P_j)\cap V(G_{2j})=\emptyset$,
  and~$V(P_i)\subseteq V(G_{2i-2})\subseteq V(G_{2j})$. 
  Thus~$V(P_i)\cap V(P_j)=\emptyset$ for~$i\neq j$.
  
  \mbox{}
\end{proof}

\vspace{-3mm}
It is easy to see that the strategy of splitter is efficiently
computable, as it amounts to computing breadth-first
searches in the subgraphs arising in the game. The splitter
game allows to recursively decompose local neighborhoods
such that the recursion tree has bounded depth. This can 
be used for example to solve the generalized distance-$r$
independent set problem that arises as a problem in the
model-checking algorithm. 

\smallskip
For the general model-checking
problem we still have to deal with two combinatorial 
problems. The first problem is the following. In a naive
approach we would translate an input formula~$\varphi$
into Gaifman normal form and for each of the local 
formulas $\chi^{(r)}(x)$ and for each vertex $v\in V(G)$
try to evaluate whether~$G\models \chi^{(r)}(v)$. 
This is equivalent to evaluating whether~$G[N_r(v)]\models
\chi^{(r)}(v)$. We would treat the $r$-neighborhood of 
each vertex $v$ as the first
move of connector in the splitter game and delete splitter's
answer from $G[N_r(v)]$. By marking the neighbors of 
all deleted vertices we can translate the formula $\chi$
to an equivalent formula $\chi'$ over an extended vocabulary. 
We then translate~$\chi'$ again into Gaifman normal form 
and recurse. The first problem of this approach 
is that when translating $\varphi$ into
Gaifman normal form, we introduce new quantifiers to 
syntactically localize the formula~$\chi^{(r)}$. This leads to a higher
locality radius $r'$ when translating the
formula~$\chi'$ again into Gaifman normal form, and so on. 
Hence, we cannot play the splitter game with the constant
radius $r$ in this naive approach. 
The second problem is that even if we fixed the first
problem the resulting algorithm would
have a worst-case running time of
$n^{\Oof(\ell(r))}$, as we create a recursion tree with 
worst-case branching degree~$n$ and depth $\ell(r)$. This
is no improvement over the simple algorithm running in time
$n^{\Oof(|\varphi|)}$. 

\smallskip
The first problem is handled as follows. We know that
the new quantifiers that are used in $\chi$ are only used
to localize the formula, that is, to express distance constraints. 
We can therefore enrich first-order logic by atoms to express
distances, so that we do not waste quantifiers for localization. 
We have to be careful though, as these new quantifiers bring
additional power to our formulas. The clue is to define a new
rank function (instead of quantifier rank) that limits the use
of distance atoms in the scope of quantifiers. Intuitively, 
the more quantifiers are available in a subformula (of original
first-order logic), the larger 
distances the formula can express. 
By carefully choosing the rank function 
we get a modified version of Gaifman's locality theorem 
such that the rank remains stable under localization. 
%I refer
%to~\cite{grohe2017deciding} for the (slightly technical) details. 

The second problem is handled as follows. We cannot afford
a branching degree $n$ in the recursion, but instead we 
must group closeby vertices that share many vertices in their
\mbox{$r$-neighborhoods} in clusters. This concept is captured 
by the notion of neighborhood covers that is explained next. 

\section{Neighborhood covers and weak coloring numbers}

The existence of \emph{sparse neighborhood covers} for 
nowhere dense graph classes is derived 
from a second characterization of treedepth via \emph{elimination
orderings}. The appropriate local version of this measures
leads to the definition of \emph{weak coloring numbers}. 
Let me define sparse neighborhood covers first. 

\begin{definition}
For~$r\in\N$, an \emph{$r$-neighborhood cover}~$\XXX$ of a graph~$G$ is a set
of connected subgraphs of~$G$ called \emph{clusters}, such that for every
vertex~$v\in V(G)$ there is some~$X\in\XXX$ with~$N_r(v)\subseteq V(X)$.
The \emph{radius}~$\rad(\XXX)$ of a cover~$\XXX$ is the maximum radius of any of
its clusters. The \emph{degree}~$d^\XXX(v)$ of~$v$ in~$\XXX$ is the number of
clusters that contain~$v$. A class $\Cc$ \emph{admits 
sparse neighborhood
covers} if there exists $c\in \N$ and for all
$r\in\N$ and all real $\epsilon>0$ a number $d=d(r,\epsilon)$ 
such that every $n$-vertex graph $G\in \Cc$ 
admits an $r$-neighborhood cover of radius at most $c\cdot
r$ and degree at most $d\cdot n^\epsilon$. 
\end{definition} 

%We furthermore say that a class
%admits \emph{efficient sparse neighborhood covers} if it admits
%sparse neighborhood covers and there exists an algorithm 
%and a computable function $f$ that on input $G,r,\epsilon$
%computes an $r$-neighborhood cover as above in time 
%$f(r,\epsilon)\cdot n^p$ for some constant $p$. 

\begin{theorem}[\cite{grohe2017deciding,grohe2015colouring}]\label{thm:nd-nc}
A class $\Cc$ is nowhere dense if and only
if the class $\Cc_\subseteq=\{H\subseteq G : G\in \Cc\}$ 
admits sparse neighborhood
covers.
\end{theorem}

The proof of the theorem is based on a characterization of 
nowhere dense classes in terms of weak coloring numbers, 
which can be seen as another local version of treedepth. 
An order of the vertex set $V(G)=\{v_1,\ldots,v_n
\}$ of an $n$-vertex graph $G$ is a permutation 
$\pi=(v_1,\ldots, v_n)$. We say that $v_i$ is smaller than $v_j$
and write $v_i<_\pi v_j$ if $i<j$. We write $\Pi(G)$ for the 
set of all orders of $V(G)$. The \emph{coloring number} $\col(G)$ of a graph 
$G$ is the minimum integer $k$ such that there exists a linear 
order $\pi$ of the vertices of $G$, such that every vertex $v$ has 
back-degree at most~$k-1$, i.e., at most $k-1$ neighbors $u$ with 
$u<_\pi v$. The coloring number of $G$ minus one is equal to the 
\emph{degeneracy} of $G$, which is the minimum integer $\ell$ such that every subgraph $H\subseteq G$
has a vertex of degree at most $\ell$.

\begin{definition}
Let $G$ be a graph and let $\pi$ be 
an order of $V(G)$. We say that a 
vertex $u\in V(G)$ is \emph{weakly reachable} with 
respect to $\pi$ from a vertex $v\in V(G)$ if $u\leq_\pi v$ and 
there exists a path~$P$ between $u$ and $v$ with 
$w>_\pi u$ for all internal vertices $w\in V(P)$. We write $
\WReach[G,\pi,v]$ for the set of vertices that are
weakly reachable from $v$. The \emph{depth} 
of~$\pi$ on $G$ is the maximum over all vertices $v$ of
$G$ of $|\WReach[G,\pi,v]|$. 
\end{definition}

\begin{lemma}[see e.g.~\cite{nevsetril2012sparsity}, Lemma 6.5]
Let $G$ be a graph. The treedepth of $G$ is equal to 
the minimum depth over all orders $\pi$ of $V(G)$. 
\end{lemma}

We can naturally define a local version of 
weak reachability. 

\begin{definition}
Let $G$ be a graph and $r\in \N$. Let $\pi$ be a 
linear order of~$V(G)$. We say that a
vertex $u\in V(G)$ is \emph{weakly $r$-reachable} with respect
to $\pi$ from a vertex
$v\in V(G)$ if $u\leq_\pi v$ and there exists a path~$P$
between $u$ and $v$ of length at most~$r$ with $w>_\pi u$ for all
internal vertices $w\in V(P)$. The set of vertices weakly $r$-reachable
by $v$ with respect to the order $\pi$ is denoted $\WReach_r[G,\pi,v]$.
We define 
\[\wcol_r(G,\pi)\coloneqq \max_{v\in V(G)}|\WReach_r[G,\pi,v]|,\]
and the \emph{weak $r$-coloring number} $\wcol_r(G)$ as
\[\wcol_r(G)\coloneqq \min_{\pi\in \Pi(G)}\max_{v\in V(G)}|\WReach_r[G,\pi,v]|.\]
\end{definition}

It is immediate from the definitions that 
\[\col(G)=\wcol_1(G)\leq \wcol_2(G)\leq \ldots \leq \wcol_n(G)=
\td(G).\] 
Hence, the weak $r$-coloring numbers can be seen as gradations 
between the coloring number~$\col(G)$ and the treedepth $\td(G)$ of
$G$. 
The weak $r$-coloring numbers 
capture local separation properties of $G$ as follows. 

\begin{lemma}\label{lem:wcol-sep}
Let $G$ be a graph, let $\pi$ be an order of $V(G)$ and let $r\in \N$. Let $u,v\in 
V(G)$, say $u<_\pi v$, and assume that $u\not\in \WReach_r[G,\pi,v]$. Then every path $P$ of length at most 
$r$ connecting $u$ and $v$ intersects $\WReach_r[G,\pi,v]\cap
\WReach_r[G,\pi,u]$. 
\end{lemma}
\begin{proof}
Let $P$ be any path of length
at most $r$ connecting $u$ and $v$. Then the minimum 
vertex of $P$ lies both in $\WReach_r[G,\pi,v]$ and in 
$\WReach_r[G,\pi,u]$. 
\end{proof}

\begin{theorem}[\cite{zhu2009colouring}]\label{thm:be-wcol}
 A class $\Cc$ of graphs has bounded expansion if and only
 if for every $r\in \N$ there exists a number $w=w(r)$ 
 such that for every $G\in \Cc$ we have 
 $\wcol_r(G)\leq w$.  
\end{theorem}

\begin{theorem}[\cite{zhu2009colouring,
nevsetvril2011nowhere}]\label{thm:nd-wcol}
 A class $\Cc$ of graphs is nowhere dense if and only
 if for every $r\in \N$ and every real $\epsilon>0$ there exists
 a number $w=w(r,\epsilon)$ such that for every $H\subseteq G\in \Cc$ we have $\wcol_r(H)\leq w\cdot |V(H)|^\epsilon$.  
\end{theorem}

To get used to the weak coloring numbers let us make the
connection with the splitter game. 

\begin{theorem}[\cite{KreutzerPRS16}]\label{thm:splitterwcol} Let $G$ be a graph, let $r\in\N$ and let
  $\ell=\wcol_{2r}(G)$. Then splitter wins the $\ell$-round radius-$r$ splitter
  game on $G$.
\end{theorem}
\begin{proof} Let $\pi$ be a linear order with 
$\WReach_{2r}[G,\pi,v]\leq \ell$ for all $v\in V(G)$. Suppose in
  round $i+1\leq \ell$, connector chooses a subgraph $C_{i+1}$ of $G_i$ of radius at most $r$. Let
  $w_{i+1}$ (splitter's choice) be the minimum vertex of 
  $C_{i+1}$
  with respect to $\pi$. Then for each $u\in V(C_{i+1})$ there is a path
  between $u$ and~$w_{i+1}$ of length at most $2r$ that uses only vertices of
  $C_{i+1}$. As~$w_i$ is minimum in $C_{i+1}$, $w_{i+1}$
  is weakly $2r$-reachable from each $u\in V(C_{i+1})$. Now let
  $G_{i+1}:=C_{i+1}-\{w_{i+1}\}]$. As~$w_{i+1}$ is not
  part of $G_{i+1}$, in the next round splitter will choose another vertex which
  is weakly $2r$-reachable from every vertex of the remaining
  graph. As $\WReach_{2r}[G,\pi,v]\leq \ell$ for all $v\in V(G)$, the game must stop after at
  most $\ell$ rounds.
\end{proof}

This gives for example a cubic number of rounds for splitter
to win on planar graphs~\cite{van2017generalised}. 
Not surprisingly, the weak coloring numbers can also be 
used to give much improved bounds for uniform quasi-wideness
on bounded expansion classes. 

\begin{theorem}[\cite{NadaraPRRS18}]\label{thm:uqw-tgv}
Let $G$ be a graph, $A\subseteq V(G)$, $r,m\in \N$ and 
assume $\wcol_r(G)=c$. If
$|A| \geq 4\cdot (2cm)^c$,  
then there exist sets $S \subseteq V(G)$
and $B \subseteq A \setminus S$ such that $|S| \leq c$, $|B| \geq m$, and~$B$ 
is $r$-independent in $G-S$.
\end{theorem}

We now come to the proof of \cref{thm:nd-nc}, which 
follows from \cref{thm:nd-wcol}
and the following lemma. 

\begin{lemma}[\cite{grohe2017deciding}]\label{thm:covers}
  Let~$G$ be a graph such that~$\mathrm{wcol}_{2r}(G)\leq s$ and
  let~$\pi$ be an order witnessing this. For $v\in V(G)$, let $m(v)$ be
  the minimum of $N_r(v)$ with respect to $\pi$. For each 
  $v\in V(G)$ let 
  \[X_{2r}[G,\pi,v] := \{w\in V(G) : v\in \WReach_{2r}[G,\pi,w]\}.\]
  Then~$\XXX\coloneqq \{X_{2r}[G,\pi,m(v)] : v\in V(G)\}$ 
  is an~$r$-neighborhood cover of~$G$ with radius at
  most~$2r$ and maximum degree at most~$s$.
\end{lemma}
\begin{proof}
  Clearly the radius of each cluster is at most~$2r$, because if~$v$
  is weakly~$2r$-reachable from~$w$, then~$w\in
  N_{2r}(v)$. Furthermore, for $v\in V(G)$ we have $N_r(v)\subseteq X_{2r}[G,\pi,m(v)]$. 
  To see this, let~$m(v)$ be the minimum
  of~$N_r(v)$ with respect to~$\pi$. Then~$m(v)$ is weakly~$2r$-reachable
  from every~$w\in N_r(v)\setminus\{m(v)\}$ as there is a path from~$w$
  to~$m(v)$ which uses only vertices of~$N_r(v)$ and has length at
  most~$2r$ and~$m(v)$ is the minimum element
  of~$N_r(v)$. Thus~$N_r(v)\subseteq X_{2r}[G,\pi,m(v)]$. Finally observe
  that for every~$v\in V(G)$,
 \begin{align*}
   d^\mathcal{X}(v) &= |\{u\in V(G) : v\in X_{2r}[G,\pi,u]\}|\\ &=
   |\{u\in V(G) : u\in \WReach_{2r}[G,\pi, v]\}| = |\WReach_{2r}[G,\pi,v]|\leq s.
 \end{align*}
\end{proof}

Observe that the above defined neighborhood cover $\Xx$ of an 
$n$-vertex graph may have $n$ elements, as there
may be one cluster for every vertex. Hence, when branching
over the elements of the cover we may have a branching 
degree of $n$. However, the degree of the
cover allows to bound the sum of all graphs in the recursion 
tree by $\Oof(n^{1+\epsilon})$ for nowhere dense classes. 
A different view on covers that leads to a smaller 
branching degree can be obtained as follows (we would branch 
over the $N$ subgraphs instead of over the $n$ clusters). 

 \begin{theorem}[\cite{patrice-sebi}]
 		Let $G$ be a graph and let $r\in\mathbb N$. Then there exist
	$N\leq {\mathrm wcol}_{4r+1}(G)$ induced subgraphs $H_1,\dots,H_N$ of $G$ such that
	\begin{enumerate}
		\item\label{it:r} for every $v\in V(G)$ there is some $1\leq i\leq N$ with $N_r(v)\subseteq H_i$;
		\item\label{it:rr} every connected component of the $H_i$'s has radius at most $2r$.
	\end{enumerate}
 \end{theorem}
\begin{proof}
  Let $\pi$ be a linear order of $V(G)$ witnessing 
  that $\wcol_{4r+1}(G)\leq N$ and let $c\colon V(G)\rightarrow \{1,\ldots,N\}$ be a coloring so that $c(u)\neq c(v)$ if 
  $u\in \WReach_{4r+1}[G,\pi,v]$. Such a coloring can be 
  computed by a simple greedy procedure. 
Let $H_i$ be the subgraph of $G$ induced by the sets
$X_{2r}[G,\pi,u]$ for all vertices $u$ with $c(u)=i$, 
where $X_{2r}[G,\pi,u]$ is defined as in \cref{thm:covers}. 
Let us show that the $H_i$ have the desired properties. 
	
	As in the proof of \cref{thm:covers} consider $v\in V(G)$ and let $m(v)$ be the minimum vertex of~$N_r(v)$ with
        respect to $\pi$. Then~$m(v)$ is weakly $2r$-reachable from every vertex in $N_r(v)$ thus \mbox{$N_r(v)\subseteq H_{c(m(v))}$} and \eqref{it:r}
        holds.
	
	As observed before, for every $v\in V(G)$ we have $X_{2r}[G,\pi,v]\subseteq
        N_{2r}(v)$. Now assume towards a contradiction that there exist $u_1<_\pi u_2$,
        $z_1\in X_{2r}[G,\pi,u_1]$, and $z_2\in X_{2r}[G,\pi,u_2]$ such that $c(u_1)=c(u_2)$ and
        $z_1$ and $z_2$ are either equal or adjacent.  Then, considering a path
        of length at most $2r$ linking $u_1$ and $z_1$ with minimum $u_1$, the
        edge $\{z_1,z_2\}$ if $z_1\neq z_2$ and a path of length at most $2r$
        linking $z_2$ and $z_2$ with minimum~$u_2$, we obtain a path of length
        at most $4r+1$ linking~$u_1$ and $u_2$ with minimum $u_1$. Hence $u_1$
        is weakly $(4r+1)$-reachable from $u_2$, contradicting the hypothesis
        $c(u_1)=c(u_2)$. It follows that all connected components of $H_i$ are
        of the form $X_{2r}[G,\pi,v]$ for some $v\in V(G)$ 
        hence have radius at
        most~$2r$. Thus \eqref{it:rr} holds.
\end{proof}

Without going into more details: neighborhood covers can 
now be used to group vertices appropriately and to efficiently
solve the model-checking problem. Further applications of the
weak coloring numbers are in the efficient approximation of the 
distance-$r$ dominating set problem~\cite{AmiriMRS18,Dvorak13,dvovrak2019distance}, as well 
as in the kernelization of distance-$r$ dominating set and 
distance-$r$ independent set~\cite{DrangeDFKLPPRVS16,EickmeyerGKKPRS17,MS19}.

\section{Conclusion and outlook}

Nowhere dense graph classes have a rich algorithmic theory
and in particular, under the assumption of subgraph closure, 
these classes constitute the border of tractability for 
first-order model-checking. Current research follows two 
lines to extend this border of tractability beyond subgraph
closed graph classes. The first line aims to study classes 
that are obtained as first-order interpretations or 
transductions of bounded expansion or nowhere dense 
classes. For example one obtains the class of map graphs
as a first-order transduction from the class of planar graphs. 
Classes that are obtained as first-order transductions of 
sparse graph classes are called \emph{structurally sparse}
in \cite{GajarskyKNMPST18}. It is a natural conjecture 
that good algorithmic 
properties of structurally sparse classes are inherited from 
the sparse base classes. I refer to \cite{KwonPS17,GajarskyK18,GajarskyHOLR16,
GajarskyKNMPST18,nesetril2019classes} for progress 
in this direction. 

The second line of research is motivated by the observation
that nowhere dense graph classes are \emph{monadically 
stable}~\cite{AdlerA14}, a property that is studied in model theory, see e.g.~\cite{baldwin1985second}. Model theory
offers a wealth of tools that could be exploited in an 
algorithmic context. For example, we proved in~\cite{FabianskiPST19} that
the distance-$r$ dominating set problem is fixed-parameter
tractable on every class of graphs where the distance-$r$
formula is both stable and equational.

\bibliographystyle{abbrv}
\bibliography{ref} 

\end{document}